\documentclass[letterpaper, 10 pt, conference]{ieeeconf}  

\IEEEoverridecommandlockouts                              

\usepackage{subfigure}

\usepackage{amsmath,amssymb,amsfonts}
\usepackage{mathrsfs}
\usepackage{multicol,lipsum}
\usepackage{mathtools, cuted,color}

\usepackage{amsthm}
\usepackage{enumerate}  
\theoremstyle{plain}
\newtheorem{theorem}{\textbf{Theorem}}[section]
\theoremstyle{plain}
\newtheorem{definition}{\textbf{Definition}}[section]
\theoremstyle{plain}
\newtheorem{lemma}{\textbf{Lemma}}[section]

\usepackage{cite}

\usepackage{algorithmic}
\usepackage[ruled,vlined]{algorithm2e}

\usepackage{graphicx}

\ifodd 0
\newcommand{\rev}[1]{{\color{blue}#1}} 
 
\newcommand{\com}[1]{\textbf{\color{red}(COMMENT: #1)}} 
\newcommand{\mcom}[1]{\textbf{\color{purple}(MahdiCom: #1)}} 
\newcommand{\clar}[1]{\textbf{\color{green}(NEED CLARIFICATION: #1)}}

\else
\newcommand{\rev}[1]{#1}

\newcommand{\com}[1]{}
\newcommand{\mcom}[1]{}

\newcommand{\clar}[1]{}

\fi
\title{\LARGE \bf
Recycled ADMM: Improve Privacy and Accuracy with Less Computation in Distributed Algorithms}

\author{Xueru Zhang, Mohammad Mahdi Khalili, Mingyan Liu
\thanks{\rev{This work is supported by the NSF under grants CNS-1422211, CNS-1646019, and CNS-1739517. }}
\thanks{X. Zhang, M. Khalili and M. Liu are with the Dept. of Electrical Engineering and Computer Science, University of Michigan, Ann Arbor, MI 48105, \{xueru, khalili, mingyan\}@umich.edu}}

\begin{document}

\maketitle
\thispagestyle{empty}
\pagestyle{empty}

\begin{abstract}
Alternating direction method of multiplier (ADMM) is a powerful method to solve decentralized convex optimization problems. In distributed settings, each node performs computation with its local data and the local results are exchanged among neighboring nodes in an iterative fashion. During this iterative process the leakage of data privacy arises and can  accumulate significantly over many iterations, making it difficult to balance the privacy-utility tradeoff.  In this study we propose Recycled ADMM (R-ADMM), where a linear approximation is applied to every even iteration, its solution directly calculated using only results from the previous, odd iteration. It turns out that under such a scheme, half of the updates incur no privacy loss and require much less computation compared to the conventional ADMM.  We obtain a sufficient condition for the convergence of R-ADMM and provide the privacy analysis based on objective perturbation.
\end{abstract}

\section{Introduction}\label{sec:intro}
Distributed optimization and learning are crucial for many settings where the data is possessed by multiple parties or when the quantity of data prohibits processing at a central location\cite{extra1,extra2,extra3,extra4,extra5}. 
Many problems can be formulated as a convex optimization of the following form: $\min_\textbf{x}\sum_{i=1}^{N}f_i(\textbf{x})$.
In a distributed setting, each entity/node $i$ has its own local objective $f_i$, $N$ entities/nodes collaboratively work to solve this objective through an interactive process of local computation and message passing. At the end all local results should ideally converge to the global optimum.

The information exchanged over the iterative process gives rise to privacy concerns if the local training data contains sensitive information such as medical or financial records, web search history, and so on. It is therefore highly desirable to ensure such iterative processes are privacy-preserving. We adopt the $\varepsilon$-differential privacy to measure such privacy guarantee; it is generally achieved by perturbing the algorithm such that the probability distribution of its output is relatively insensitive to any change to a single record in the input \cite{Dwork2006}.

Existing approaches to decentralizing the above problem primarily consist of subgradient-based algorithms \cite{nedic2008,nedic2009,lobel2011,Gade} and ADMM-based algorithms \cite{xu2017adaptive,xu2016,zhang2017privacy,zhang2017admm,zhang2018distributed,wei2012,ling2014,shi2014,zhang2014,ling2016,Chunlei}. It has been shown that ADMM-based algorithms can converge at the rate of $O(\frac{1}{k})$ while subgradient-based algorithms typically converge at the rate of $O(\frac{1}{\sqrt{k}})$, where $k$ is the number of iterations \cite{wei2012}. In this study, we will solely focus on ADMM-based algorithms. While a number of differentially private (sub)gradient-based distributed algorithms have been proposed \cite{hale2015,huang2015,han2017,bellet2017}, the same is much harder for ADMM-based algorithms due to its computational complexity stemming from the fact that each node is required to solve an optimization problem in each iteration. To the best of our knowledge, only \cite{zhang2017,xueru,DP_ADMM} apply differential privacy to ADMM. In particular, \cite{zhang2017} proposed the dual/primal variable perturbation method to inspect the privacy loss of one node in every single iteration; this, however, is not sufficient for guaranteeing privacy as an adversary can potentially use the revealed results from all iterations to perform inference. In \cite{xueru} we address this issue by inspecting the total privacy loss over the entire process and the entire network; we proposed a penalty perturbation method which improves the privacy-utility tradeoff significantly. \rev{In the recent work \cite{DP_ADMM}, a first-order approximation is applied to the augmented Lagrangian in all iterations; however, this method requires a central server to average all updated primal variables over the network in each iteration.}

In this study we present Recycled ADMM (R-ADMM), a modified version of ADMM where the privacy leakage only happens during half of the updates. Specifically, we adopt a linearized approximated optimization in every even iteration, whose solution can actually be calculated directly from results in the previous, odd iteration, and is used for updating primal variable. We establish a sufficient condition for convergence and provide a privacy analysis using the objective perturbation method. Our numerical results show that the privacy-utility tradeoff can be improved significantly. 

The remainder of the paper is organized as follows. We present problem formulation and definition of differential privacy and ADMM in Section \ref{sec:pre} and the Recycled ADMM algorithm along with its convergence analysis in Section \ref{sec:radmm}.  A private version of this ADMM algorithm is then introduced in Section \ref{sec:pradmm} and numerical results in Section \ref{sec:numerical}. Section \ref{sec:conclusion} concludes the paper.

\section{Preliminaries }\label{sec:pre}
\subsection{Problem Formulation}
Consider a connected network\footnote{A connected network is one in which every node is reachable (via a path) from every other node.} given by an undirected graph $G(\mathscr{N},\mathscr{E})$, which consists of a set of nodes $\mathscr{N} = \{1,2,\cdots,N\}$ and a set of edges $\mathscr{E} = \{1,2,\cdots,E\}$. Two nodes can exchange information if and only if they are connected by an edge. Let $\mathscr{V}_i$ denote node $i$'s set of neighbors, excluding itself. A node $i$ has a dataset $D_i = \{(x_{i}^n,y_{i}^n) | n = 1,2,\cdots,B_i \}$, where $x_{i}^n \in \mathbb{R}^d$ is the feature vector representing the $n$-th sample belonging to $i$, $y_{i}^n \in \{-1,1\}$ the corresponding label, and $B_i$ the size of $D_i$. 

Consider the regularized empirical risk minimization (ERM) problem for binary classification defined as follows:
\begin{equation}\label{eq:prelimi_1}
\min_{f_c}O_{ERM}(f_{c},D_{all}) = \sum_{i=1}^{N}\dfrac{C}{B_i}\sum_{n=1}^{B_i} {\mathscr{L}}(y_{i}^n f_c^Tx_{i}^n ) + \rho R(f_c) 
\end{equation}
where $C \leq B_i$ and $\rho>0$ are constant parameters of the algorithm, the loss function $\mathscr{L}(\cdot)$ measures the accuracy of the classifier, and the regularizer $R(\cdot)$ helps prevent overfitting. The goal is to train a (centralized) classifier $f_c\in \mathbb{R}^d$ over the union of all local datasets $D_{all} = \cup_{i \in \mathscr{N}} D_i$ in a distributed manner using ADMM, while providing privacy guarantee for each data sample.
\subsection{Differential Privacy \cite{Dwork2006}}
A randomized algorithm $ \mathscr{A}(\cdot)$ taking a dataset as input satisfies $\varepsilon$-differential privacy if for any two datasets $D$, $\hat{D}$ differing in at most one data point, and for any set of possible outputs $S \subseteq \text{range}(\mathscr{A})$, $\text{Pr}(\mathscr{A}(D) \in S)\leq e^{\varepsilon}\text{Pr}(\mathscr{A}(\hat{D}) \in S)$ holds.
We call two datasets differing in at most one data point as neighboring datasets. The above definition suggests that for a sufficiently small $\varepsilon$, an adversary will observe almost the same output regardless of the presence (or value change) of any one individual in the dataset; this is what provides privacy protection for that individual. 
\subsection{Conventional ADMM}
To decentralize \eqref{eq:prelimi_1}, let $f_i$ be the local classifier of each node $i$. To achieve consensus, i.e., $f_1 = f_2 = \cdots = f_N$, a set of auxiliary variables $\{w_{ij} | i \in \mathscr{N}, j \in \mathscr{V}_i\}$ are introduced for every pair of connected nodes.  As a result, \eqref{eq:prelimi_1} is reformulated equivalently as: 
 \begin{equation}\label{eq:prelimi_2}
\begin{aligned}
& \min_{\{f_i\},\{w_{ij}\}} 
& &\tilde{O}_{ERM}(\{f_i\}_{i=1}^N,D_{all})  = \sum_{i=1}^{N}O(f_i,D_i) \\
&\text{   s.t.} 
& & f_i = w_{ij}, w_{ij} = f_j, \ \ \ i \in \mathscr{N}, j \in \mathscr{V}_i
\end{aligned}
\end{equation}
where $O(f_i,D_i) = \dfrac{C}{B_i}\sum_{n=1}^{B_i} {\mathscr{L}}(y_{i}^n f_i^Tx_{i}^n ) + \dfrac{\rho}{N} R(f_i)$. $\{f_i\}$ (resp. $\{w_{ij}\}$) is the shorthand for $\{f_i\}_{i\in \mathscr{N}}$ (resp. $\{w_{ij}\}_{i\in \mathscr{N}, j\in \mathscr{V}_i}$). Let $\{w_{ij},\lambda_{ij}^k\}$ be the shorthand for $\{w_{ij},\lambda_{ij}^k\}_{i \in \mathscr{N},j \in \mathscr{V}_i, k \in \{a,b\}}$, where $\lambda_{ij}^a$, $\lambda_{ij}^b$ are dual variables corresponding to equality constraints $f_i = w_{ij}$ and $w_{ij}=f_j$ respectively. The objective in \eqref{eq:prelimi_2} can be solved using ADMM with the augmented Lagrangian: 
\begin{eqnarray}\label{eq:prelimi_3}
L_\eta (\{f_i\},\{w_{ij},\lambda_{ij}^k\}) = \sum_{i=1}^NO(f_i,D_i)\nonumber\\ +  \sum_{i=1}^N\sum_{j \in \mathscr{V}_i}(\lambda_{ij}^a)^T(f_i-w_{ij})+\sum_{i=1}^N\sum_{j \in \mathscr{V}_i}(\lambda_{ij}^b)^T(w_{ij}-f_j)\\
+ \sum_{i=1}^N\sum_{j \in \mathscr{V}_i}\dfrac{\eta}{2} (||f_i-w_{ij}||_2^2+||w_{ij}-f_j||_2^2)~.\nonumber
\end{eqnarray}

In the $(t+1)$-th iteration, the ADMM updates consist of the following:
\begin{eqnarray}
f_i(t+1) = \underset{f_i}{\text{argmin}}\ L_\eta (\{f_i\},\{w_{ij}(t),\lambda_{ij}^k(t)\})~;\label{eq:prelimi_4}\\
w_{ij}(t+1) = \underset{w_{ij}}{\text{argmin}}\ L_\eta (\{f_i(t+1)\},\{w_{ij},\lambda_{ij}^k(t)\})\label{eq:prelimi_5}~;\\
\lambda_{ij}^a(t+1) = \lambda_{ij}^a(t) + \eta(f_i(t+1)-w_{ij}(t+1))~;\label{eq:prelimi_6}\\
\lambda_{ij}^b(t+1) = \lambda_{ij}^b(t) + \eta(w_{ij}(t+1)-f_j(t+1))~.\label{eq:prelimi_7}
\end{eqnarray}
Using Lemma 3 in \cite{forero2010}, {if dual variables $\lambda_{ij}^a(t)$ and $\lambda_{ij}^b(t)$ are initialized to zero for all node pairs $(i,j)$, then $\lambda_{ij}^a(t) = \lambda_{ij}^b(t)$ and $\lambda_{ij}^k(t) = -\lambda_{ji}^k(t)$ will hold for all iterations with $k \in \{a,b\}, i \in \mathscr{N}, j \in \mathscr{V}_i$.} Let $\lambda_{i}(t) = \sum_{j \in \mathscr{V}_i}\lambda_{ij}^a(t) = \sum_{j \in \mathscr{V}_i}\lambda_{ij}^b(t)$, then the ADMM iterations \eqref{eq:prelimi_4}-\eqref{eq:prelimi_7} can be simplified as (Refer to Appendix A in \cite{xueru} for proof):
\begin{eqnarray}
f_i(t+1) = \underset{f_i}{\text{argmin}} \{ O(f_i,D_i) + 2\lambda_i(t)^Tf_i \nonumber \\
+  \eta \sum_{j \in \mathscr{V}_i}||\dfrac{1}{2}(f_i(t)+f_j(t))-f_i||_2^2~ \}~; \label{eq:prelimi_8} \\ 
\lambda_{i}(t+1) = \lambda_{i}(t) +  \dfrac{\eta}{2}\sum_{j \in \mathscr{V}_i}(f_i(t+1)-f_j(t+1))~. \label{eq:prelimi_9}
\end{eqnarray}

 \subsection{Private ADMM \cite{zhang2017} \& Private M-ADMM \cite{xueru}}
 In private ADMM \cite{zhang2017}, the noise is added either to the updated primal variable before broadcasting to its neighbors (primal variable perturbation), or to the dual variable before updating its primal variable using \eqref{eq:prelimi_8} (dual variable perturbation). The privacy property is only evaluated for a single node and a single iteration, both methods cannot balance the privacy-utility tradeoff very well if consider the total privacy loss. In \cite{xueru} the total privacy loss of the whole network over the entire iterative process is considered. A modified ADMM (M-ADMM) was proposed to improve the privacy-utility tradeoff. Specifically, it explores the rule of step-size (penalty parameter) in stabilizing the algorithm. M-ADMM allows each node to independently determine its penalty parameter; by perturbing the algorithm with noise correlated to penalty parameter and at the same time increasing the penalty parameters, the privacy and accuracy can be improved simultaneously.    
 
\subsection{Main idea} 
Fundamentally, the accumulation of privacy loss over iterations stems from the fact that the raw data is used in every primal update. If the updates can be made without using the raw data, but only from computational results that already exist, then the privacy loss originating from these updates will be zero, while at the same time the computational cost be reduced significantly. Based on this idea, we start with modifying ADMM such that we can repeatedly use some computational results to make updates. 
 
\section{Recycled ADMM (R-ADMM)} \label{sec:radmm} 

\subsection{Making information recyclable} 
ADMM can outperform gradient-based methods in terms of requiring fewer number of iterations for convergence; this however comes at the price of high computational cost in every iteration. In particular, the primal variable is updated by performing an optimization in each iteration. In \cite{mokhtari2015,ling2014,ling2015dlm}, either a linear or  quadratic approximation of the objective function is used to obtain an inexact solution in each iteration in lieu of solving the original optimization problem. While this clearly lowers the computational cost, the approximate computation is performed using the local, raw data in every iteration, which means that privacy loss inevitably accumulates over the iterations. 

We begin by modifying ADMM in such a way that in every even iteration, without using the raw data, the primal variable is updated solely based on the existing computational results from the previous, odd iteration. Compared with conventional ADMM, these updates incur no privacy loss and less computation.  Since the computational results are repeatedly used, this method will be referred to as Recycled ADMM (R-ADMM). 

Specifically, in the $2k$-th (even) iteration, we approximate $O(f_i,D_i)$ (Eqn. \eqref{eq:prelimi_8}, primal update optimization) by $O(f_i,D_i)\approx O(f_i(2k-1),D_i) + \nabla O(f_i(2k-1),D_i)^T(f_i-f_i(2k-1)) + \frac{\gamma}{2}||f_i-f_i(2k-1)||_2^2$ $(\gamma \geq 0)$ and update only the primal variables. Using the first-order condition, the updates in the $2k$-th iteration become:
\begin{eqnarray}
f_i(2k)=f_i(2k-1) - \frac{1}{2\eta V_i+\gamma}\{\nabla O(f_i(2k-1),D_i) \nonumber \\ +2\lambda_i(2k-1)+\eta\sum_{j\in \mathscr{V}_i}(f_i(2k-1)-f_j(2k-1))\}
\label{eq:r-admm1}~;\\
\lambda_{i}(2k) = \lambda_{i}(2k-1)~. \label{eq:r-admm2}
\end{eqnarray}
In the $(2k-1)$-th (odd) iteration, the updates are kept the same as \eqref{eq:prelimi_8}\eqref{eq:prelimi_9}:
\begin{eqnarray}
f_i(2k-1) = \underset{f_i}{\text{argmin}} \{ O(f_i,D_i) + 2\lambda_i(2k-2)^Tf_i \nonumber \\
+  \eta \sum_{j \in \mathscr{V}_i}||\dfrac{1}{2}(f_i(2k-2)+f_j(2k-2))-f_i||_2^2~ \}~;  \label{eq:r-admm3}\\ 
\lambda_{i}(2k-1) = \lambda_{i}(2k-2) \nonumber\\+  \dfrac{\eta}{2}\sum_{j \in \mathscr{V}_i}(f_i(2k-1)-f_j(2k-1))~. \label{eq:r-admm4}
\end{eqnarray}
Note that in the $(2k)$-th (even) iteration, we need the gradient $\nabla O(f_i(2k-1),D_i)$ and primal difference $\dfrac{\eta}{2}\sum_{j \in \mathscr{V}_i}(f_i(2k-1)-f_j(2k-1))$ for the updates; these are available directly from the previous, $(2k-1)$-th (odd) iteration,
i.e., this information can be recycled.  In this sense, R-ADMM may be viewed as alternating between conventional ADMM (odd iterations) and a variant of gradient descent (even iterations), where $\frac{1}{2\eta {V}_i + \gamma}$ is the step-size with \rev{a slightly modified gradient term}. The procedure is shown in Algorithm \ref{A1}.
\begin{algorithm}\label{A1}
	\textbf{Input: }{$\{D_i\}_{i=1}^N$}
	
	\textbf{Initialize: }$\forall i$, generate $f_i(0)$ randomly, $\lambda_i(0) = \textbf{0}_{d \times 1}$ 
	
	\For{$k=1$ \KwTo$K$}{
		\For{$i=1$ \KwTo$\mathscr{N}$}{
			Update primal variable $f_i(2k-1)$ via \eqref{eq:r-admm3};
			
			Calculate the gradient $\nabla O(f_i(2k-1),D_i)$;
			
			Broadcast $f_i(2k-1)$ to all neighbors $j \in \mathscr{V}_i$.
		}	
		
		\For{$i=1$ \KwTo$\mathscr{N}$}{
			Calculate  $\dfrac{\eta}{2}\sum_{j \in \mathscr{V}_i}(f_i(2k-1)-f_j(2k-1))$;
			
			Update dual variable $\lambda_i(2k-1)$ via \eqref{eq:r-admm4}.
		}	
		\For{$i=1$ \KwTo$\mathscr{N}$}{
			Use the stored $\nabla O(f_i(2k-1),D_i)$ and $\dfrac{\eta}{2}\sum_{j \in \mathscr{V}_i}(f_i(2k-1)-f_j(2k-1))$ to update primal variable $f_i(2k)$ via \eqref{eq:r-admm1};
			
			Keep the dual variable $\lambda_i(2k)=\lambda_i(2k-1)$;
			
			Broadcast $f_i(2k)$ to all neighbors $j \in \mathscr{V}_i$.
		}	
	}
	\textbf{Output: }{primal $\{f_i(2K)\}_{i=1}^N$ and dual $\{\lambda_i(2K)\}_{i=1}^N$}	\caption{Recycled ADMM (R-ADMM)}
\end{algorithm}

\subsection{Convergence Analysis}
We next show that R-ADMM (Eqn. \eqref{eq:r-admm1}-\eqref{eq:r-admm4})  converges to the optimal solution under a set of common technical assumptions. 

\textbf{\textit{Assumption 1}:} Function $O(f_i,D_i)$ is convex and differentiable in $f_i$, $\forall i$.

\textbf{\textit{Assumption 2}:} The solution set to the original ERM problem \eqref{eq:prelimi_1} is nonempty and there exists at least one bounded element. 

\textbf{\textit{Assumption 3:}} For all $i \in \mathscr{N}$, $O(f_i,D_i)$ has Lipschitz continuous gradients, i.e., for any $f_i^1$ and ${f}_i^2$, we have: 
\begin{equation}\label{eq:assume1}
||\nabla O(f_i^1,D_i)-\nabla O(f_i^2,D_i)||_2 \leq M_i||f_i^1- f_i^2||_2
\end{equation}

By the KKT condition of the primal update \eqref{eq:r-admm3}:
\begin{eqnarray}\label{eq:c_1}
0 = \nabla O(f_i(2k-1),D_i) + 2\lambda_i(2k-2)\nonumber\\ + \eta\sum_{j \in \mathscr{V}_i}(2{f}_i(2k-1)-({f}_i(2k-2)+{f}_j(2k-2)))~.
\end{eqnarray} 

Define the adjacency matrix $A\in \mathbb{R}^{N \times N}$ as: 
$$a_{ij} = \begin{cases}
1, \ \ \text{ if node } i\text{ and node }j \text{ are connected } \\
0, \ \ \text{ otherwise }~. 
\end{cases}$$

Stack the variables $f_i(t)$, $\lambda_i(t)$ and $\nabla O(f_i(t),D_i)$ for $i \in \mathscr{N}$ into matrices, i.e.,
$$ \hat{f}(t) = \begin{bmatrix}
f_1(t)^T\\
f_2(t)^T\\
\vdots\\
f_N(t)^T
\end{bmatrix}\in \mathbb{R}^{N\times d} \text{ , \ \   }\Lambda(t) = \begin{bmatrix}
\lambda_1(t)^T\\
\lambda_2(t)^T\\
\vdots\\
\lambda_N(t)^T
\end{bmatrix}\in \mathbb{R}^{N\times d} $$ 
$$ \nabla \hat{O}(\hat{f}(t),D_{all}) = \begin{bmatrix}
\nabla O(f_1(t),D_1)^T\\
\nabla O(f_2(t),D_2)^T\\
\vdots\\
\nabla O(f_N(t),D_N)^T
\end{bmatrix}\in \mathbb{R}^{N\times d} $$

Let $V_i = | \mathscr{V}_i|$ be the number of neighbors of node $i$, and define the degree matrix $D = \textbf{diag}([V_1;V_2;\cdots;V_N])\in \mathbb{R}^{N \times N}$ and the diagonal matrix $\tilde{D}$ with $\tilde{D}_{ii} = 2\eta V_i + \gamma$. Then for each $k$, the matrix form of \eqref{eq:r-admm1}\eqref{eq:r-admm2}\eqref{eq:c_1}\eqref{eq:r-admm4} are:
\begin{eqnarray}
\hat{f}(2k) =\hat{f}(2k-1) - \tilde{D}^{-1}\{\nabla \hat{O}(\hat{f}(2k-1),D_{all})\nonumber \\+2\Lambda(2k-1)
+ \eta(D-A)\hat{f}(2k-1)\}~;\label{eq:matrix_1}\\
2\Lambda(2k) = 2\Lambda(2k-1)~;\label{eq:matrix_2}\\
\textbf{0}_{N\times d}=\nabla\hat{O}(\hat{f}(2k-1),D_{all}) + 2\Lambda(2k-2)\nonumber \\+2\eta D\hat{f}(2k-1)
-\eta (D+A)\hat{f}(2k-2)~;\label{eq:matrix_3}\\
2\Lambda(2k-1) = 2\Lambda(2k-2)+\eta(D-A)\hat{f}(2k-1)~.\label{eq:matrix_4}
\end{eqnarray}
Writing $\hat{f}(2k-2)$ and $\Lambda(2k-2)$ in \eqref{eq:matrix_3}\eqref{eq:matrix_4} as functions of $\hat{f}(2k-3)$, $\Lambda(2k-3)$ using \eqref{eq:matrix_1}\eqref{eq:matrix_2}, we obtain: 
\begin{eqnarray*}
\nabla\hat{O}(\hat{f}(2k-1),D_{all}) +\eta (D+A) \tilde{D}^{-1}\nabla \hat{O}(\hat{f}(2k-3),D_{all}) \nonumber \\+\eta(D+A)(\hat{f}(2k-1)-\hat{f}(2k-3))
\\+ \eta(D+A)\tilde{D}^{-1} \eta(D-A)\hat{f}(2k-3)\nonumber\\ +2\Lambda(2k-1)+\eta (D+A) \tilde{D}^{-1}2\Lambda(2k-3)=\textbf{0}_{N\times d}~;\nonumber \\
2\Lambda(2k-1) = 2\Lambda(2k-3)+\eta(D-A)\hat{f}(2k-1)\nonumber~.
\end{eqnarray*}	
The convergence of R-ADMM is proved by showing that the pair ($\hat{f}(2k-1)$, $\Lambda(2k-1)$) from odd iterations converges to the optimal solution. To simplify the notation, we will re-index every two consecutive odd iterations $2k-3$ and $2k-1$ using $t$ and $t+1$:
\begin{eqnarray}
\nabla\hat{O}(\hat{f}(t+1),D_{all}) +\eta (D+A) \tilde{D}^{-1}\nabla \hat{O}(\hat{f}(t),D_{all}) \nonumber \\+\eta(D+A)((\hat{f}(t+1)-\hat{f}(t))
+ \tilde{D}^{-1} \eta(D-A)\hat{f}(t))\nonumber\\ +2\Lambda(t+1)+\eta (D+A) \tilde{D}^{-1}2\Lambda(t)=\textbf{0}_{N\times d}~;\label{eq:c_2} \\
2\Lambda(t+1) = 2\Lambda(t)+\eta(D-A)\hat{f}(t+1)\label{eq:c_3}~.
\end{eqnarray}

Note that $D-A$ is the laplacian and $D+A$ is the signless Laplacian matrix of the network, with the following properties if the network is connected: {(i)} $D\pm A \succeq 0$ is positive semi-definite; {(ii)} $\text{Null}(D-A) = c\textbf{1}$, i.e., every member in the null space of $D-A$ is a scalar multiple of \textbf{1} with \textbf{1} being the vector of all $1$'s \cite{Jonathan2007}. 
\begin{lemma}\label{Lemma:1}
	[\textbf{First-order Optimality Condition} \cite{ling2016}] Under Assumptions 1 and 2, the following two statements are equivalent:
	\begin{itemize}
		\item $\hat{f}^* = [(f_1^*)^T;(f_2^*)^T;\cdots;(f_N^*)^T] \in \mathbb{R}^{N\times d}$ is consensual, i.e., $f_1^*=f_2^*=\cdots=f_N^*=f_c^*$ where $f_c^*$ is the optimal solution to \eqref{eq:prelimi_1}.
		\item There exists a pair $(\hat{f}^*,\Lambda^*)$ with $2\Lambda^* = (D-A)X$ for some $X\in \mathbb{R}^{N\times d}$ such that
		\begin{eqnarray}
		\nabla \hat{O}(\hat{f}^*,D_{all})+2\Lambda^*=\textbf{0}_{N\times d} ~; 
		\label{eq:c_4}\\
		(D-A)\hat{f}^* = \textbf{0}_{N\times d}~. 
		\label{eq:c_5}
		\end{eqnarray}
	\end{itemize}
\end{lemma}
Lemma \ref{Lemma:1} shows that a pair $(\hat{f}^*,\Lambda^*)$ satisfying \eqref{eq:c_4}\eqref{eq:c_5} is equivalent to the optimal solution of our problem, hence the convergence of R-ADMM is proved by showing that $(\hat{f}(t),\Lambda(t))$ in \eqref{eq:c_2}\eqref{eq:c_3} converges to a pair $(\hat{f}^*,\Lambda^*)$ satisfying \eqref{eq:c_4}\eqref{eq:c_5}. 

\begin{theorem}\label{thmC1}[\textbf{Sufficient Condition}]
	Consider the modified ADMM defined by \eqref{eq:c_2}\eqref{eq:c_3}. Let $\{\hat{f}(t),\Lambda(t)\}$ be outputs in each iteration and $\{\hat{f}^*,\Lambda^*\}$ a pair satisfying \eqref{eq:c_4}\eqref{eq:c_5}. Denote $D_M = \textbf{diag}([M_1^2;M_2^2;\cdots;M_N^2])\in \mathbb{R}^{N \times N}$ with $0<M_i<+\infty$ as given in Assumption 3. If the following two conditions hold for some constants $L>0$ and $\mu>1$:
	\begin{eqnarray}
	(I+\eta (D+A) \tilde{D}^{-1}) \succ \frac{L\mu}{2\sigma_{\min}(\tilde{D})}\frac{1}{\eta}D_M(D-A)^{+} ~;\label{eq:c_6}\\
	\eta(D+A)\succ \{\eta (D+A) \tilde{D}^{-1} \eta(D-A)\nonumber\\ +\frac{2}{L}\eta (D+A) \tilde{D}^{-1}\eta (D+A)+\frac{L\mu}{2\sigma_{\min}(\tilde{D})(\mu-1)}D_M \}~.\label{eq:c_7}
	\end{eqnarray}
	where $\sigma_{\min}(\tilde{D}) = \min_i\{2\eta V_i+\gamma\}$ is the smallest singular value of $\tilde{D}$, then $(\hat{f}(t),\Lambda(t))$ converges to $(\hat{f}^*,\Lambda^*)$.
\end{theorem}
\begin{proof}
	\rev{See Appendix \ref{App1}.}
\end{proof}

By controlling $\gamma$, it is easy to find constants $L>0$ and $\mu>1$ such that  conditions \eqref{eq:c_6}\eqref{eq:c_7} are satisfied, and they are not unique. One example is $L=2$ and $\mu=2$, in which case \eqref{eq:c_6}\eqref{eq:c_7} are reduced to:
\begin{eqnarray}
(I+\eta (D+A) \tilde{D}^{-1}) \succ \frac{4}{2\sigma_{\min}(\tilde{D})}\frac{1}{\eta}D_M(D-A)^{+} ~;\label{eq:c_8}\\
\eta(D+A)\succ 2\eta (D+A) \tilde{D}^{-1} \eta D +\frac{2}{\sigma_{\min}(\tilde{D})}D_M ~.\label{eq:c_9}
\end{eqnarray}
\eqref{eq:c_8}\eqref{eq:c_9} can be easily satisfied for sufficiently large $\gamma\geq 0$. Note that the conditions are sufficient but not necessary, so in practice convergence may be attained under weaker settings. 

\section{Private R-ADMM}\label{sec:pradmm} 

In this section we present a privacy preserving version of R-ADMM. In odd iterations, we adopt the objective perturbation \cite{chaudhuri2011} where a random linear term $\epsilon_i(2k-1)^Tf_i$ is added to the objective function in \eqref{eq:r-admm3}
\footnote{Pure differential privacy was adopted in this work, but the weaker $(\epsilon,\delta)$-differential privacy can be applied as well.}, where $\epsilon_i(2k-1)$ follows the probability density proportional to $\exp\{-\alpha_i(k)||\epsilon_i(2k-1)||_2\}$.
\begin{eqnarray}\label{eq:P_modify_1}
{L}_i^{priv}(2k-1) = O(f_i,D_i) + (2\lambda_i(2k-2)+\epsilon_i(2k-1))^Tf_i \nonumber \\
+  \eta \sum_{j \in \mathscr{V}_i}||\dfrac{1}{2}(f_i(2k-2)+f_j(2k-2))-f_i||_2^2 \nonumber
\end{eqnarray}
To generate this noisy vector, choose the norm from the gamma distribution with shape $d$ and scale $\frac{1}{\alpha_i(k)}$ and the direction uniformly, where $d$ is the dimension of the feature space. Node $i$'s local result is obtained by finding the optimal solution to the private objective function: 
\begin{equation}\label{eq:P_modify_2}
f_i(2k-1) = \underset{f_i}{\text{argmin}}\ {L}_i^{priv}(2k-1) , \ \ i \in \mathscr{N}~. 
\end{equation}
\rev{In the $2k$-th iteration, use the stored results $\epsilon_i(2k-1) + \nabla O(f_i(2k-1),D_i)$ and $\dfrac{\eta}{2}\sum_{j \in \mathscr{V}_i}(f_i(2k-1)-f_j(2k-1))$ to update primal variables, where the latter can be obtained from the dual update in the $(2k-1)$-th update, and the former can be obtained directly from the KKT condition in the $(2k-1)$-th iteration:
	\begin{eqnarray*}\label{info}
\epsilon_i(2k-1) + \nabla O(f_i(2k-1),D_i) = -2\lambda_{i}(2k-2)\nonumber \\-\eta\sum_{j \in \mathscr{V}_i}(2f_i(2k-1))-f_i(2k-2)-f_j(2k-2)) ~. 
	\end{eqnarray*}
Then the even update is given by: 
\begin{eqnarray}\label{eq:P_modify_3}
f_i(2k)=f_i(2k-1) - \frac{1}{2\eta V_i+\gamma}\{2\lambda_i(2k-1)\nonumber \\+\underbrace{\epsilon_i(2k-1) +\nabla O(f_i(2k-1),D_i)}_{\text{the existing result by KKT} } \nonumber \\+\underbrace{\eta\sum_{j\in \mathscr{V}_i}(f_i(2k-1)-f_j(2k-1))}_{\text{the existing result by the previous dual update}}\}~. 
\end{eqnarray}}
 Algorithm \ref{A2} shows the complete procedure, where the condition used to generate $\eta$ helps to bound the worst-case privacy loss but is not necessary in guaranteeing convergence.
\begin{algorithm}\label{A2}
	\textbf{Input: }{$\{D_i\}_{i=1}^N$}, $\{\alpha_i(1),\cdots, \alpha_i(K)\}_{i=1}^N$
	
	\textbf{Initialize: }$\forall i$, generate $f_i(0)$ randomly, $\lambda_i(0) = \textbf{0}_{d \times 1}$ 
	
	\textbf{Parameter: }Select $\eta$ s.t. $2c_1<\min_i\{\frac{B_i}{C}(\frac{\rho}{N} + 2\eta V_i)\}$
	
	\For{$k=1$ \KwTo$K$}{
		\For{$i=1$ \KwTo$\mathscr{N}$}{
			Generate noise $\epsilon_i(2k-1) \sim \exp(-\alpha_i(k)||\epsilon||_2)$;
			
			Update primal variable $f_i(2k-1)$ via \eqref{eq:P_modify_2};
			
			
			Broadcast $f_i(2k-1)$ to all neighbors $j \in \mathscr{V}_i$.
		}	
		
		\For{$i=1$ \KwTo$\mathscr{N}$}{
			Calculate  $\dfrac{\eta}{2}\sum_{j \in \mathscr{V}_i}(f_i(2k-1)-f_j(2k-1))$;
			
			Update dual variable $\lambda_i(2k-1)$ via \eqref{eq:r-admm4}.
		}	
		\For{$i=1$ \KwTo$\mathscr{N}$}{
			Use the stored information $\epsilon_i(2k-1)+\nabla O(f_i(2k-1),D_i)$ and $\dfrac{\eta}{2}\sum_{j \in \mathscr{V}_i}(f_i(2k-1)-f_j(2k-1))$ to update primal variable $f_i(2k)$ via \eqref{eq:P_modify_3};
			
			Keep the dual variable $\lambda_i(2k)=\lambda_i(2k-1)$;
			
			Broadcast $f_i(2k)$ to all neighbors $j \in \mathscr{V}_i$.
		}	
	}
	\textbf{Output: }{Upper bound of the total privacy loss $\beta$; primal $\{f_i(2K)\}_{i=1}^N$ and dual $\{\lambda_i(2K)\}_{i=1}^N$}	
	\caption{Private R-ADMM}
\end{algorithm}

In the distributed and iterative setting, the ``output'' of the algorithm is not merely the end result, but includes all intermediate results generated and exchanged during the iterative process. For this reason, we adopt the differential privacy definition proposed in \cite{xueru} as follows.  
\begin{definition}\label{Def}	
	Consider a connected network $G(\mathscr{N},\mathscr{E})$ with a set of nodes $\mathscr{N} = \{1,2,\cdots,N\}$. Let $f(t) = \{f_i(t)\}_{i=1}^N$ denote the information exchange of all nodes in the $t$-th iteration.
	A distributed algorithm is said to satisfy $\beta$-differential privacy during $T$ iterations if for any two datasets $D_{all} = \cup_i D_i$ and $\hat{D}_{all} = \cup_i \hat{D}_i$, differing in at most one data point, and for any set of possible outputs $S$ during $T$ iterations, the following holds:
	\begin{equation*}
	\frac{\text{Pr}(\{f(t)\}_{t=0}^T \in S|D_{all})}{\text{Pr}(\{f(t)\}_{t=0}^T \in S|\hat{D}_{all})} \leq \exp(\beta)
	\end{equation*}
\end{definition}

We now state another result of this paper, on the privacy property of the private R-ADMM (Algorithm \ref{A2}) using the above definition. Additional assumptions on $\mathscr{L}(\cdot)$ and $R(\cdot)$ are used. 

\textbf{\textit{Assumption 4}:} The loss function $\mathscr{L}$ is strictly convex and twice differentiable. $|\nabla\mathscr{L}| \leq 1$ and $0 <\mathscr{L}''\leq c_1$ with $c_1$ being a constant. 

\textbf{\textit{Assumption 5}:} The regularizer $R$ is 1-strongly convex and twice continuously differentiable. 
\begin{lemma}\label{lemmaP1}
	Consider the private R-ADMM (Algorithm \ref{A2}), $\forall k=1,\cdots K$, assume the total privacy loss up to the $(2k-1)$-th iteration can be bounded by $\beta_{2k-1}$, then the total privacy loss up to the $2k$-th iteration can also be bounded by $\beta_{2k-1}$. In other words, given the private results in odd iterations, outputting private results in the even iterations does not release more information about the input data.  
\end{lemma}
\begin{proof}
	\rev{See Appendix \ref{App_2}.}
\end{proof}

\begin{theorem}\label{thmP}
	Normalize feature vectors in the training set such that $||x_{i}^n||_2\leq 1$ for all $i \in \mathscr{N}$ and $n$. Then the private R-ADMM algorithm (Algorithm 2) satisfies the $\beta$-differential privacy with 
	\begin{equation}
	\beta \geq \underset{i \in \mathscr{N}}{\max}\{\sum_{k=1}^{K}\frac{2C}{B_i}(\frac{1.4c_1}{(\frac{\rho}{N}+2\eta V_i)} + \alpha_i(k))\}~. 
	\end{equation}
\end{theorem}
\begin{proof}
\rev{See Appendix \ref{App_3}.}
\end{proof}
\begin{figure}
	\centering   
	{\includegraphics[width=0.49\textwidth]{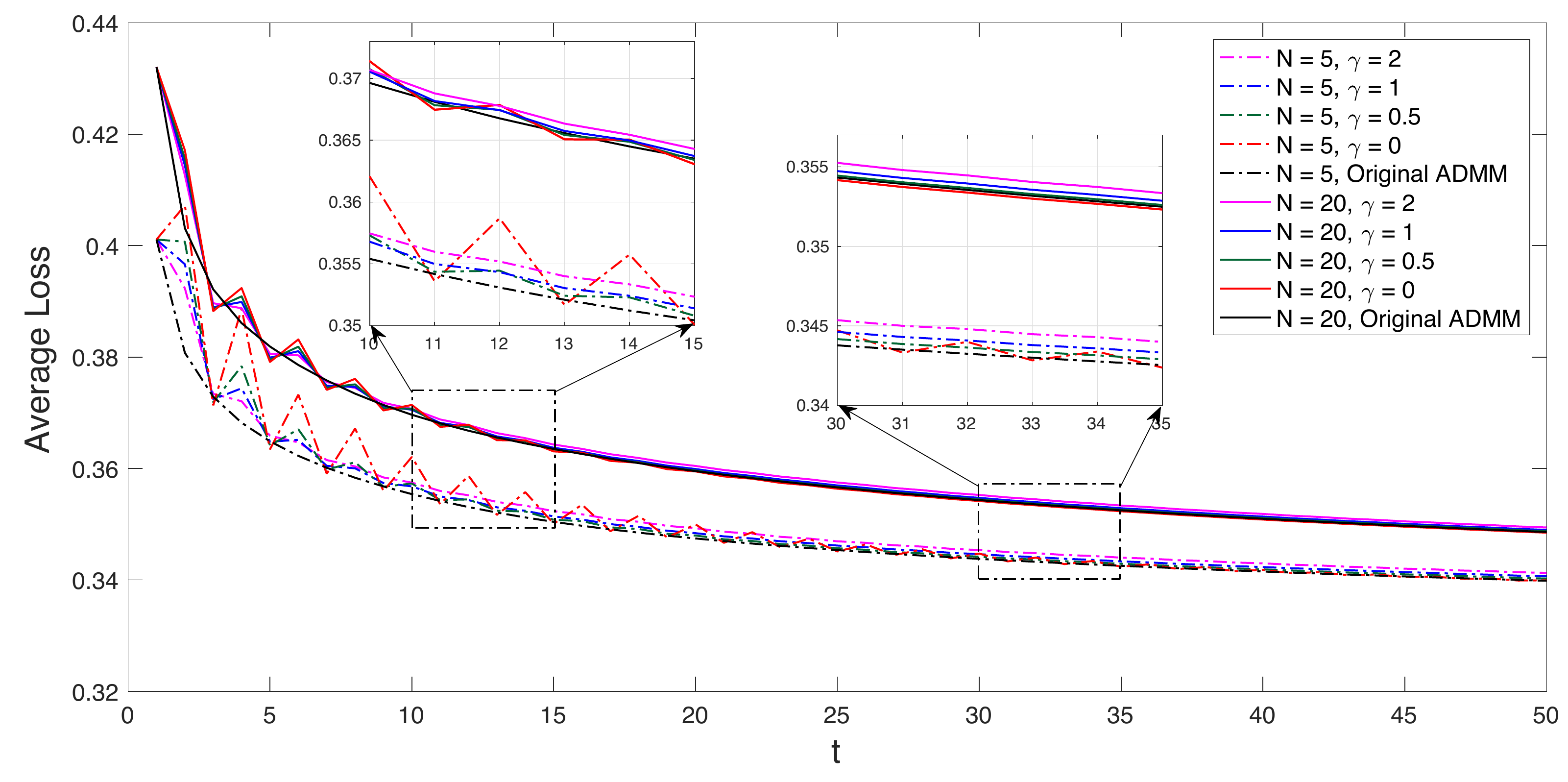}}
	\caption{Convergence properties of R-ADMM.}\label{fig1}
\end{figure}
\section{Numerical Experiments}\label{sec:numerical} 
We use the \textit{Adult} dataset from the UCI Machine Learning Repository \cite{Lichman2013}. It consists of personal information of around 48,842 individuals, including age, sex, race, education, occupation, income, etc. The goal is to predict whether the annual income of an individual is above \$50,000. 

Following the same pre-processing steps as in \cite{xueru}, the final data includes 45,223 individuals, each represented as a 105-dimensional vector of norm at most 1.
We will use as loss function the logistic loss $\mathscr{L}(z) = \log(1+\exp(-z))$, with $|\mathscr{L}'|\leq 1 $ and $\mathscr{L}'' \leq c_1 = \frac{1}{4}$. 
The regularizer is $R(f_i) = \frac{1}{2}||f_i||_2^2$. 
We will measure the accuracy of the algorithm by the average loss $L(t):=\frac{1}{N} \sum_{i=1}^{N}\frac{1}{B_i}\sum_{n=1}^{B_i}\mathscr{L}(y^n_if_i(t)^Tx^n_i) $ over the training set. We will measure the privacy of the algorithm by the upper bound $P(t):=\underset{i \in \mathscr{N}}{\max}\{\sum_{k=1}^{K}\frac{2C}{B_i}(\frac{1.4c_1}{(\frac{\rho}{N}+2\eta V_i)} + \alpha_i(k))\}$. 
The smaller $L(t)$ and $P(t)$, the higher accuracy and stronger privacy guarantee.


\subsection{Convergence of non-private R-ADMM}
\begin{figure}
	\centering   
	\subfigure[Accuracy comparison: $\alpha=2$]{\label{fig2:a}\includegraphics[trim={3cm 0 3cm 1cm},clip=true, width=0.47\textwidth]{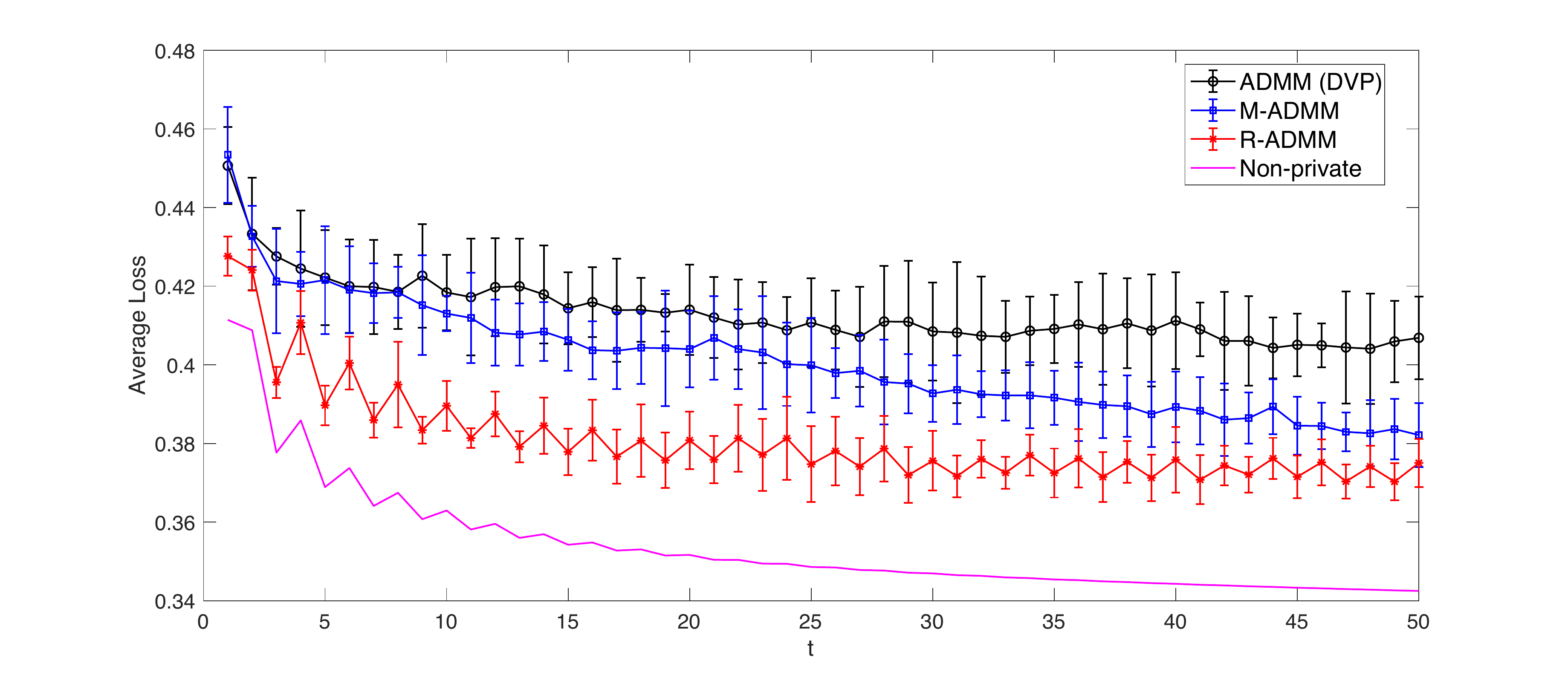}}
	\subfigure[Accuracy comparison: $\alpha=4$]{\label{fig2:b}\includegraphics[trim={3cm 0 3cm 1cm},clip=true,width=0.47\textwidth]{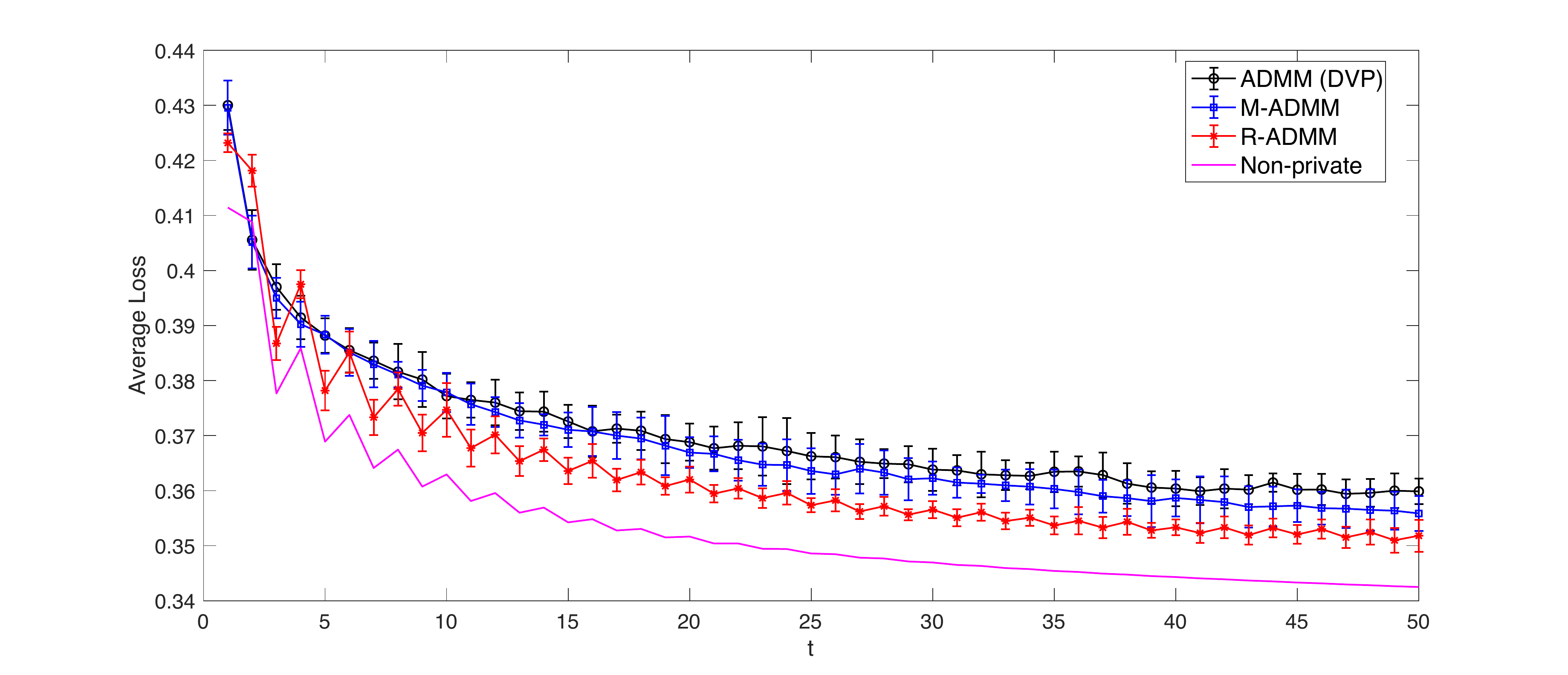}}
	\subfigure[Privacy comparison: $\alpha=2$]{\label{fig2:c}\includegraphics[trim={1.3cm 0 1.3cm 1cm},clip=true,width=0.23\textwidth]{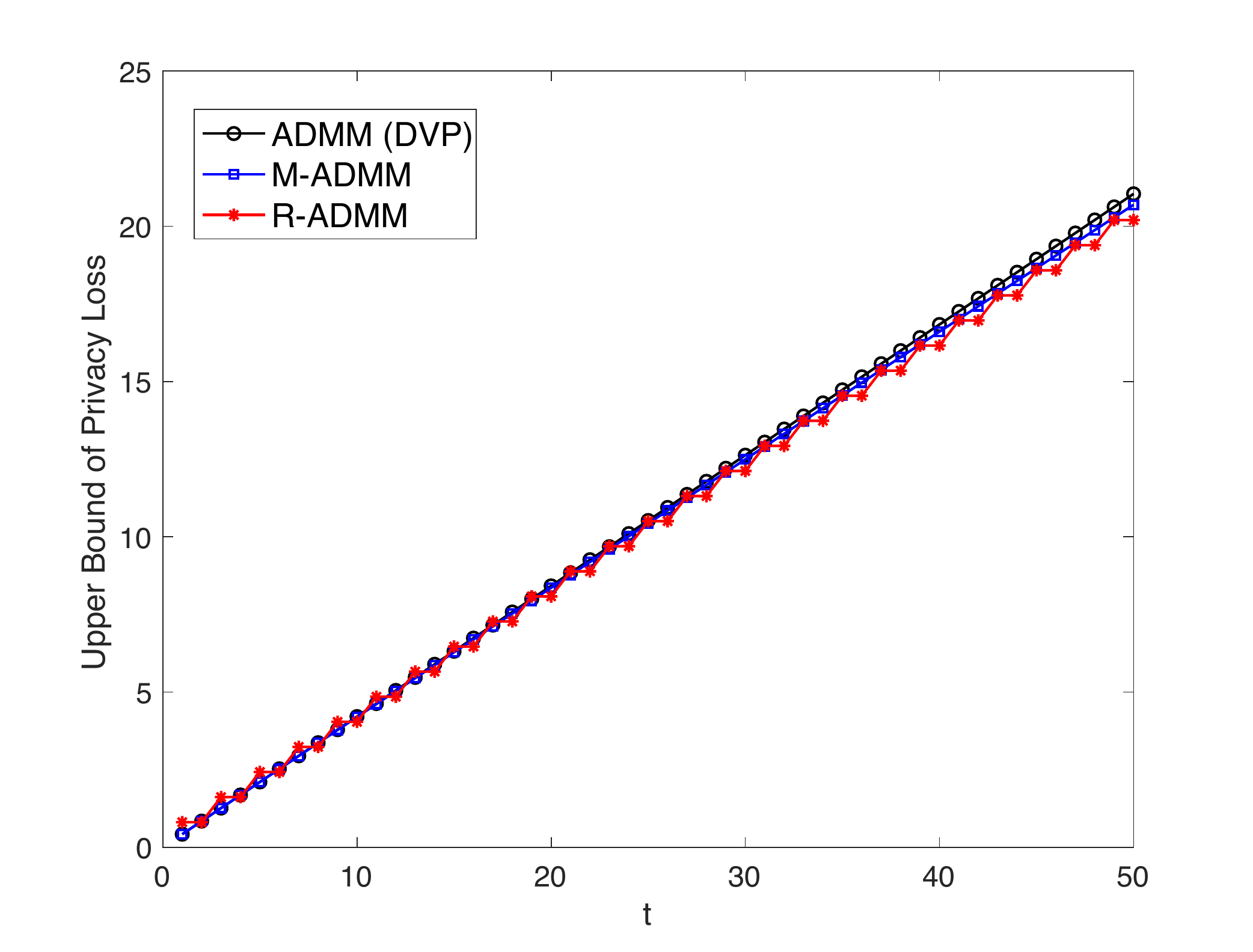}}
	\subfigure[Privacy comparison:  $\alpha=4$]{\label{fig2:d}\includegraphics[trim={1.3cm 0 1.3cm 1cm},clip=true,width=0.23\textwidth]{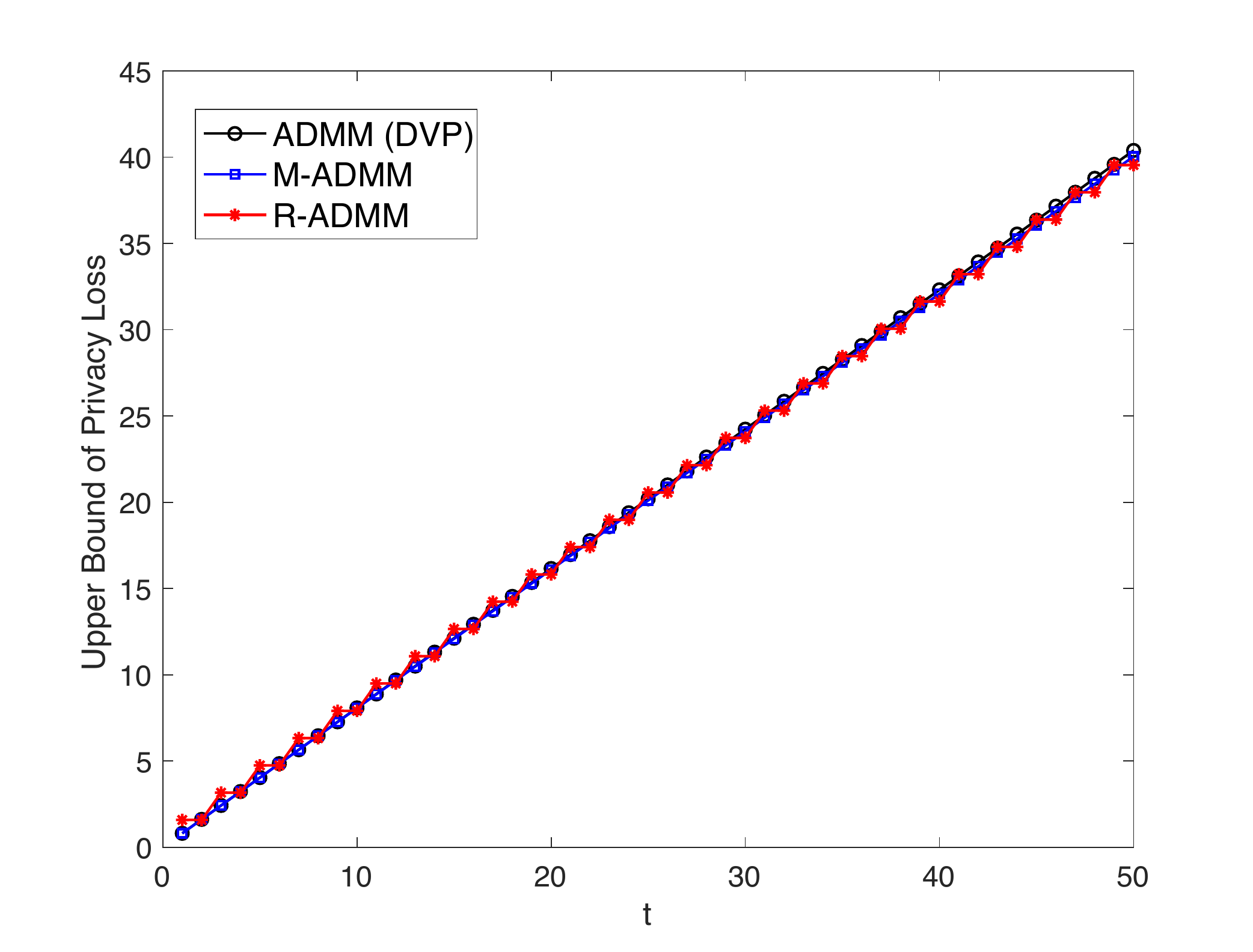}}
	\caption{Comparison of accuracy and privacy.}
	\label{fig2}
\end{figure}
Figure \ref{fig1} shows the convergence of R-ADMM with different $\gamma$ and fixed $\eta=0.5$ for a small network ($N=5$) and a large network ($N=20$), both are randomly generated. Due to the linear approximation in even iterations, it's possible to cause an increased average loss as shown in the plot. However, the odd iterations will always compensate this increase; if we only look at the odd iterations, R-ADMM achieves a similar convergence rate as conventional ADMM. $\gamma$ can also be thought of as an extra penalty parameter for each node in even iterations to punish its update, i.e., the difference between $f_i(2k)$ and $f_i(2k-1)$.  Larger $\gamma$ can result in smaller oscillation between even and odd iterations but will also lower the convergence rate.  
\subsection{Private R-ADMM}
\begin{figure}
	\centering   
	{\includegraphics[width=0.47\textwidth]{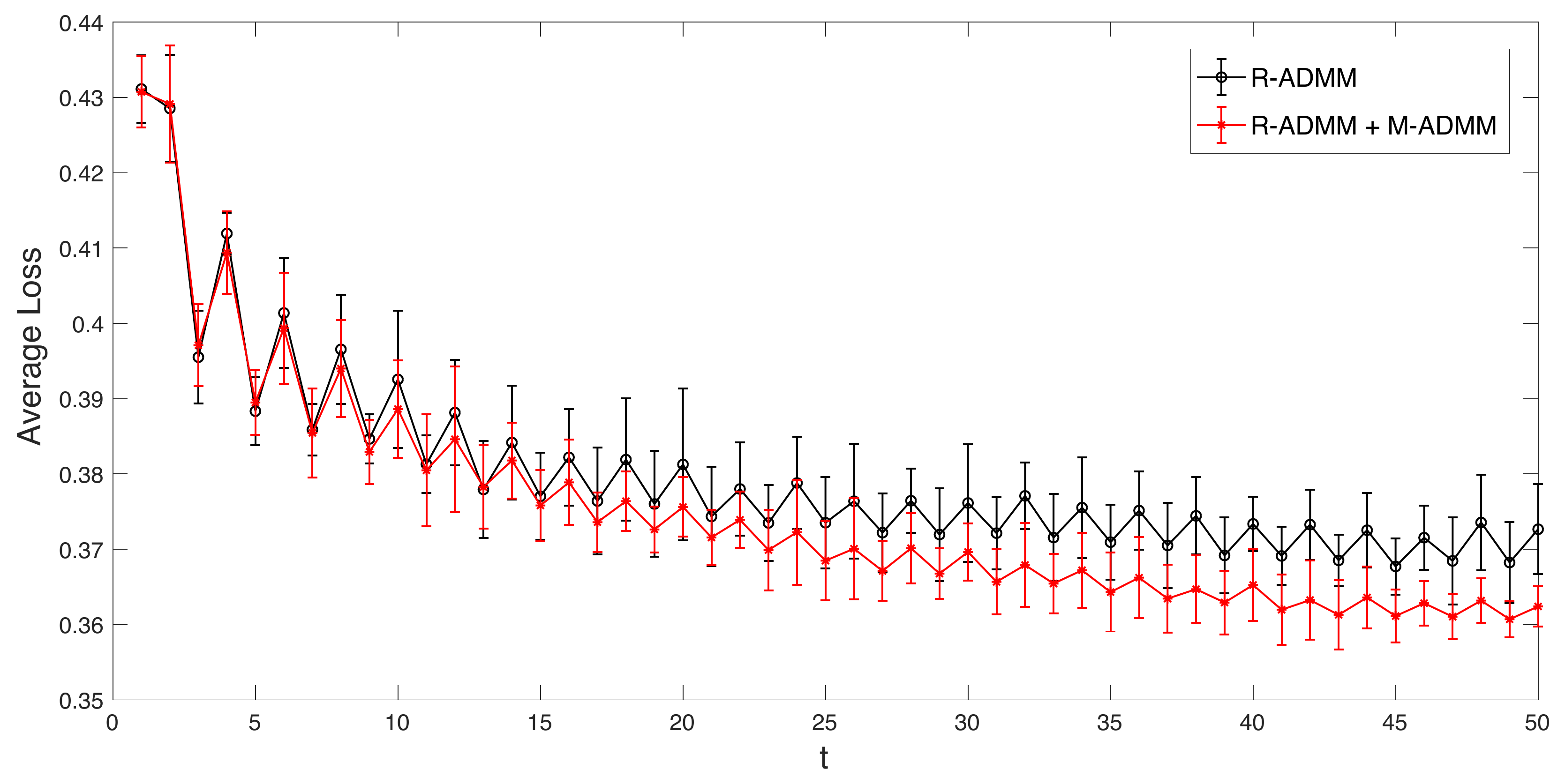}}
	\caption{Accuracy comparison: $\eta(t) = 1.01^t$, $\gamma(t) = 0.2*1.01^t$}\label{fig3}
\end{figure}
We next inspect the accuracy and privacy of the private R-ADMM (Algorithm \ref{A2}) and compare it with the private (conventional) ADMM using dual variable perturbation (DVP) \cite{zhang2017} and the private M-ADMM using penalty perturbation (PP) \cite{xueru}. In the set of experiments, we fix $\gamma = 0.2$, $\eta=1$ in private R-ADMM and set the noise parameter $\alpha_i(k) =\alpha, \forall i,k$. The noise parameters of conventional ADMM and M-ADMM are also chosen respectively such that they have almost the same total privacy loss bounds. 

For each parameter setting, we perform 10 independent runs of the algorithm, and record both the mean and the range of their accuracy.  Specifically, $L^l(t)$  denotes the average loss over the training dataset in the $t$-th iteration of the $l$-th experiment ($1\leq l \leq 10$). The mean of average loss is then given by $L_{mean}(t) = \frac{1}{10}\sum_{l=1}^{10} L^l(t)$, and the range $L_{range}(t) = \underset{1\leq l \leq 10}{\max} L^l(t) - \underset{1\leq l \leq 10}{\min} L^l(t)$. 
The larger the range $L_{range}(t)$ the less stable the algorithm, i.e., under the same parameter setting, the difference in performances (convergence curves) of every two experiments is larger. 
Each parameter setting also has a corresponding upper bound on the privacy loss denoted by $P(t)$. 
Figures \ref{fig2:a}-\ref{fig2:b} show both $L_{mean}(t)$ and $L_{range}(t)$ as vertical bars centered at $L_{mean}(t)$.  Their corresponding privacy upper bound is given in Figures \ref{fig2:c}-\ref{fig2:d}. The pair \ref{fig2:a}, \ref{fig2:c} (resp. \ref{fig2:b}, \ref{fig2:d}) is for the same parameter setting. 
We see that the private R-ADMM (red) has higher accuracy than both the private ADMM (black) and M-ADMM (blue), and the improvement is more significant with the smaller total privacy loss.

We also incorporate the idea from \cite{xueru} into private R-ADMM, where we decrease the step-size, i.e., increase $\eta$ and $\gamma$, over iterations to stabilize the algorithm and improve the algorithmic performance. The result is shown in Figure \ref{fig3}  where the privacy loss bound is controlled to be the same during the whole period. It shows that by varying the step-size, the privacy-utility tradeoff can be further improved.

\section{Conclusion}\label{sec:conclusion}
We presented Recycled ADMM (R-ADMM), a modified version of ADMM that can improve the privacy-utility tradeoff significantly with less computation. The idea is to repeatedly use the existing computational results instead of the raw data to make updates. We also established a sufficient condition for convergence and privacy analysis using objective perturbation.


\vfill 
\appendices
\begin{figure*}[t]
	\normalsize
	\begin{eqnarray}\label{eq:thmC1}
	\langle \hat{f}(t+1)-\hat{f}^*,-\eta (D+A) \tilde{D}^{-1}(\nabla \hat{O}(\hat{f}(t),D_{all}) -\nabla \hat{O}(\hat{f}^*,D_{all}))+(I+\eta (D+A) \tilde{D}^{-1})(2\Lambda^*-2\Lambda(t+1))\nonumber \\+\eta (D+A) \tilde{D}^{-1}(2\Lambda(t+1)-2\Lambda(t))-\eta(D+A)(\hat{f}(t+1)-\hat{f}(t))
	-\eta (D+A) \tilde{D}^{-1} \eta(D-A)\hat{f}(t) \rangle_F \geq 0 ~.
	\end{eqnarray}
	\hrulefill
	\begin{eqnarray}
	\langle \hat{f}(t+1)-\hat{f}^*,\eta (D+A) \tilde{D}^{-1}(2\Lambda(t+1)-2\Lambda(t)) -\eta (D+A) \tilde{D}^{-1} \eta(D-A)\hat{f}(t) \rangle_F\nonumber\\
	=\langle \hat{f}(t+1)-\hat{f}^*,\eta (D+A) \tilde{D}^{-1} \eta(D-A)(\hat{f}(t+1)-\hat{f}(t)) \rangle_F\label{eq:thmC2} \\= \frac{1}{2}||\hat{f}(t+1)-\hat{f}^*||^2_{G_1} +  \frac{1}{2}||\hat{f}(t+1)-\hat{f}(t)||^2_{G_1}- \frac{1}{2}||\hat{f}(t)-\hat{f}^*||^2_{G_1}~;\nonumber\\
	\langle \hat{f}(t+1)-\hat{f}^*, (I+\eta (D+A) \tilde{D}^{-1})(2\Lambda^*-2\Lambda(t+1)) \rangle_F
	\nonumber \\=\langle  \frac{1}{\eta}(D-A)^{+}(2\Lambda(t+1)-2\Lambda(t)), (I+\eta (D+A) \tilde{D}^{-1})(2\Lambda^*-2\Lambda(t+1)) \rangle_F
 \label{eq:thmC3}\\=\frac{1}{2}||2\Lambda^*-2\Lambda(t)||^2_{G_2}-\frac{1}{2}||2\Lambda^*-2\Lambda(t+1)||^2_{G_2} - \frac{1}{2}||2\Lambda(t+1)-2\Lambda(t)||^2_{G_2}~;\nonumber\\
	\langle \hat{f}(t+1)-\hat{f}^*,-\eta(D+A)(\hat{f}(t+1)-\hat{f}(t))\rangle_F\nonumber \\
	=\frac{1}{2}||\hat{f}(t)-\hat{f}^*||^2_{\eta(D+A)}-\frac{1}{2}||\hat{f}(t+1)-\hat{f}^*||^2_{\eta(D+A)}-\frac{1}{2}||\hat{f}(t)-\hat{f}(t+1)||^2_{\eta(D+A)}~. \label{eq:thmC4}
	\end{eqnarray}
	\hrulefill
	\begin{eqnarray}\label{eq:thmC5}
	\langle \hat{f}(t+1)-\hat{f}^*,-\eta (D+A) \tilde{D}^{-1}(\nabla \hat{O}(\hat{f}(t),D_{all}) -\nabla \hat{O}(\hat{f}^*,D_{all}))\rangle_F
	\nonumber \\=\langle \hat{f}(t+1)-\hat{f}(t)+\hat{f}(t) -\hat{f}^*,\nonumber -\eta (D+A) \tilde{D}^{-1}(\nabla \hat{O}(\hat{f}(t),D_{all}) -\nabla \hat{O}(\hat{f}^*,D_{all}))\rangle_F
	\nonumber \\ \leq \langle \hat{f}(t)-\hat{f}(t+1),\eta (D+A) \tilde{D}^{-1}(\nabla \hat{O}(\hat{f}(t),D_{all}) -\nabla \hat{O}(\hat{f}^*,D_{all}))\rangle_F \nonumber 
	\nonumber \\=  \langle \eta (D+A) \sqrt{\tilde{D}^{-1}}(\hat{f}(t)-\hat{f}(t+1)),\sqrt{\tilde{D}^{-1}}(\nabla \hat{O}(\hat{f}(t),D_{all}) -\nabla \hat{O}(\hat{f}^*,D_{all}))\rangle_F~.
	\end{eqnarray}
	\hrulefill
	\begin{eqnarray}\label{eq:thmC8}
	\eqref{eq:thmC5} \leq \frac{1}{L}||(\hat{f}(t)-\hat{f}(t+1))||^2_{\eta (D+A) \tilde{D}^{-1}\eta (D+A)} + \frac{L}{4\sigma_{\min}(\tilde{D})} ( \mu||\hat{f}^*-\hat{f}(t+1)||^2_{D_M}+\frac{\mu}{\mu-1}||\hat{f}(t+1)-\hat{f}(t)||^2_{D_M})\nonumber \\
	= \frac{1}{2}||(\hat{f}(t)-\hat{f}(t+1))||^2_{\frac{2}{L}\eta (D+A) \tilde{D}^{-1}\eta (D+A)+\frac{L\mu}{2\sigma_{\min}(\tilde{D})(\mu-1)}D_M} + \frac{1}{2}||2\Lambda(t+1)-2\Lambda(t)||^2_{\frac{L\mu}{2\sigma_{\min}(\tilde{D})}(\frac{1}{\eta}(D-A)^{+})^2D_M} 
	\end{eqnarray}
	\hrulefill
	\begin{eqnarray}\label{eq:thmC9}
	\frac{1}{2}||\hat{f}(t)-\hat{f}(t+1)||^2_{\eta(D+A)-G_1}
	- \frac{1}{2}||(\hat{f}(t)-\hat{f}(t+1))||^2_{\frac{2}{L}\eta (D+A) \tilde{D}^{-1}\eta (D+A)+\frac{L\mu}{2\sigma_{\min}(\tilde{D})(\mu-1)}D_M} \nonumber \\+  \frac{1}{2}||2\Lambda(t+1)-2\Lambda(t)||^2_{G_2}- \frac{1}{2}||2\Lambda(t+1)-2\Lambda(t)||^2_{\frac{L\mu}{2\sigma_{\min}(\tilde{D})}(\frac{1}{\eta}(D-A)^{+})^2D_M}
	\nonumber \\ \leq \frac{1}{2}||\hat{f}(t+1)-\hat{f}^*||^2_{G_1} - \frac{1}{2}||\hat{f}(t)-\hat{f}^*||^2_{G_1} +\frac{1}{2}||2\Lambda^*-2\Lambda(t)||^2_{G_2}  \nonumber\\-\frac{1}{2}||2\Lambda^*-2\Lambda(t+1)||^2_{G_2} 
	+\frac{1}{2}||\hat{f}(t)-\hat{f}^*||^2_{\eta(D+A)}-\frac{1}{2}||\hat{f}(t+1)-\hat{f}^*||^2_{\eta(D+A)}
	\end{eqnarray}	
	\hrulefill
\end{figure*}
\section{Proof of Theorem \ref{thmC1}}\label{App1}
By convexity of $O(f_i,D_i)$, $(f_i^1-{f}^2_i)^T(\nabla O(f_i^1,D_i)-\nabla O({f}^2_i,D_i)) \geq 0$ holds $\forall$ $f_i^1, {f}_i^2$. Let $\langle\cdot,\cdot\rangle_F$ be frobenius inner product of two matrices, there is: $$\langle \hat{f}(t+1)-\hat{f}^*,\nabla \hat{O}(\hat{f}(t+1),D_{all})-\nabla \hat{O}(\hat{f}^*,D_{all})\rangle_F \geq 0$$ According to \eqref{eq:c_2}\eqref{eq:c_4} and \eqref{eq:c_3}, substitute $\nabla\hat{O}(\hat{f}(t+1),D_{all})-\nabla \hat{O}(\hat{f}^*,D_{all})$ and add an extra term $\eta (D+A)\tilde{D}^{-1}(\nabla \hat{O}(\hat{f}^*,D_{all})+2\Lambda^*)=\textbf{0}_{N\times d}$, implies Eqn. \eqref{eq:thmC1}.


To simplify the notation, for a matrix $X$, let $||X||^2_{J} = \langle X, JX \rangle_F$ and $(X)^+$ be the pseudo inverse of $X$.  Define:
\begin{eqnarray}
	G_1 = \eta (D+A) \tilde{D}^{-1} \eta(D-A) ~;\nonumber \\
	G_2 = \frac{1}{\eta}(D-A)^{+}(I+\eta (D+A) \tilde{D}^{-1})~.\nonumber 
\end{eqnarray} 
 Use \eqref{eq:c_3}\eqref{eq:c_5} and the fact that $\langle A, JB \rangle_F=\langle J^TA, B \rangle_F$, Eqn. \eqref{eq:thmC2}\eqref{eq:thmC3}\eqref{eq:thmC4} hold.
Let $\sqrt{X}$ denote the square root of a symmetric positive semi-definite (PSD) matrix $X$ that is also symmetric PSD. 
Eqn. \eqref{eq:thmC5} holds, 
where the inequality uses the facts that $O(f_i,D_i)$ is convex for all $i$ and that the matrix $\eta (D+A) \tilde{D}^{-1}$ is positive definite.

According to  \eqref{eq:assume1} in Assumption 3, 
 define the matrix $D_M = \textbf{diag}([M_1^2;M_2^2;\cdots;M_N^2])\in \mathbb{R}^{N \times N}$, it implies
$||\nabla \hat{O}(\hat{f}^1,D_{all}) -\nabla \hat{O}(\hat{f}^2,D_{all})||^2_F \leq \langle \hat{f}^1 - \hat{f}^2,D_M(\hat{f}^1 - \hat{f}^2) \rangle_F$. 
Since $\langle A, B \rangle_F\leq \frac{1}{L}||A||^2_F + \frac{L}{4}||B||_F^2$ holds for any $L>0$, there is:
\begin{eqnarray}\label{eq:thmC7}
\eqref{eq:thmC5} \leq \frac{1}{L}||\eta (D+A) \sqrt{\tilde{D}^{-1}}(\hat{f}(t)-\hat{f}(t+1))||^2_F\nonumber \\ + \frac{L}{4}||\sqrt{\tilde{D}^{-1}}(\nabla \hat{O}(\hat{f}(t),D_{all}) -\nabla \hat{O}(\hat{f}^*,D_{all}))||_F^2\nonumber \\
\leq \frac{1}{L}||(\hat{f}(t)-\hat{f}(t+1))||^2_{\eta (D+A) \tilde{D}^{-1}\eta (D+A)} \nonumber \\+ \frac{L\sigma_{\max}(\tilde{D}^{-1})}{4} ||\nabla \hat{O}(\hat{f}(t),D_{all}) -\nabla \hat{O}(\hat{f}^*,D_{all})||_F^2\nonumber \\
= \frac{1}{L}||(\hat{f}(t)-\hat{f}(t+1))||^2_{\eta (D+A) \tilde{D}^{-1}\eta (D+A)}\nonumber \\ + \frac{L}{4\sigma_{\min}(\tilde{D})} ||\hat{f}^*-\hat{f}(t)||^2_{D_M}
\end{eqnarray}
where $\sigma_{\max}(\cdot)$, $\sigma_{\min}(\cdot)$ denote the largest and smallest singular value of a matrix respectively. Since for any $\mu > 1$ and any matrices $C_1$, $C_2$, $J$ with the same dimensions, there is $||C_1+C_2||^2_J \leq \mu||C_1||^2_J+ \frac{\mu}{\mu - 1}||C_2||^2_J$. which implies:
\begin{eqnarray}
||\hat{f}^*-\hat{f}(t)||^2_{D_M} = ||\hat{f}^*-\hat{f}(t+1)+\hat{f}(t+1)-\hat{f}(t)||^2_{D_M}
\nonumber \\\leq \mu||\hat{f}^*-\hat{f}(t+1)||^2_{D_M}+\frac{\mu}{\mu-1}||\hat{f}(t+1)-\hat{f}(t)||^2_{D_M}\nonumber
\end{eqnarray}
Plug into \eqref{eq:thmC7} and use \eqref{eq:c_3}\eqref{eq:c_5} gives Eqn. \eqref{eq:thmC8}.

Combine \eqref{eq:thmC2}\eqref{eq:thmC3}\eqref{eq:thmC4}\eqref{eq:thmC8}, \eqref{eq:thmC1} becomes Eqn. \eqref{eq:thmC9}.
Suppose the following two conditions hold for some constants $L>0$ and $\mu>1$:
\begin{eqnarray}
(I+\eta (D+A) \tilde{D}^{-1}) \succ \frac{L\mu}{2\sigma_{\min}(\tilde{D})}\frac{1}{\eta}D_M(D-A)^{+} ~;\label{eq:thmC14}\\
\eta(D+A)\succ \eta (D+A) \tilde{D}^{-1} \eta(D-A)\nonumber \\ +\frac{2}{L}\eta (D+A) \tilde{D}^{-1}\eta (D+A)+\frac{L\mu}{2\sigma_{\min}(\tilde{D})(\mu-1)}D_M ~.\label{eq:thmC15}
\end{eqnarray}

Substitute $G_1 = \eta (D+A) \tilde{D}^{-1} \eta(D-A)$ and $G_2 = \frac{1}{\eta}(D-A)^{+}(I+\eta (D+A) \tilde{D}^{-1})$, define $R_1$ and $R_2$ below gives:
\begin{eqnarray}
R_1 = \eta(D+A)-G_1-\frac{L\mu}{2\sigma_{\min}(\tilde{D})(\mu-1)}D_M\nonumber \\-\frac{2}{L}\eta (D+A) \tilde{D}^{-1}\eta (D+A)\succ \textbf{0}_{N\times N} ~;\label{eq:thmC16}\\
R_2 = G_2 - \frac{L\mu}{2\sigma_{\min}(\tilde{D})}(\frac{1}{\eta}(D-A)^{+})^2D_M\succ \textbf{0}_{N\times N}~.\label{eq:thmC17}
\end{eqnarray}
Eqn. \eqref{eq:thmC9} becomes:
\begin{eqnarray}\label{eq:thmC18}
\frac{1}{2}||\hat{f}(t)-\hat{f}(t+1)||^2_{R_1}+  \frac{1}{2}||2\Lambda(t+1)-2\Lambda(t)||^2_{R_2}
\nonumber \\ \leq \frac{1}{2}||\hat{f}(t+1)-\hat{f}^*||^2_{G_1} - \frac{1}{2}||\hat{f}(t)-\hat{f}^*||^2_{G_1} \nonumber \\+\frac{1}{2}||2\Lambda^*-2\Lambda(t)||^2_{G_2}-\frac{1}{2}||2\Lambda^*-2\Lambda(t+1)||^2_{G_2} 
\nonumber \\+\frac{1}{2}||\hat{f}(t)-\hat{f}^*||^2_{\eta(D+A)}-\frac{1}{2}||\hat{f}(t+1)-\hat{f}^*||^2_{\eta(D+A)}
\end{eqnarray}
Sum up \eqref{eq:thmC18} over $t$ from $0$ to $+\infty$ leads to:
\begin{eqnarray}\label{eq:thmC19}
\sum_{t=0}^{\infty}\{||\hat{f}(t)-\hat{f}(t+1)||^2_{R_1}+  ||2\Lambda(t+1)-2\Lambda(t)||^2_{R_2}\} \nonumber \\ \leq ||\hat{f}(0)-\hat{f}^*||^2_{\eta(D+A)}-||\hat{f}(+\infty)-\hat{f}^*||^2_{\eta(D+A)}\nonumber \\ +||\hat{f}(\infty)-\hat{f}^*||^2_{G_1} - ||\hat{f}(0)-\hat{f}^*||^2_{G_1} \nonumber \\+||2\Lambda^*-2\Lambda(0)||^2_{G_2}-||2\Lambda^*-2\Lambda(\infty)||^2_{G_2}
\end{eqnarray}
The RHS of \eqref{eq:thmC19} is finite, implies that $\lim_{t\rightarrow\infty}\{||\hat{f}(t)-\hat{f}(t+1)||^2_{R_1}+  ||2\Lambda(t+1)-2\Lambda(t)||^2_{R_2}\} = 0$. Since $R_1$, $R_2$ are not unique, by \eqref{eq:thmC16}\eqref{eq:thmC17}, it requires $\lim_{t\rightarrow\infty}||\hat{f}(t)-\hat{f}(t+1)||^2_{R_1}=0$ and $\lim_{t\rightarrow\infty}||2\Lambda(t+1)-2\Lambda(t)||^2_{R_2} = 0$ should hold for all possible $R_1$, $R_2$. Therefore, $\lim_{t\rightarrow\infty}(\hat{f}(t)-\hat{f}(t+1))=\textbf{0}_{N\times d}$ and $\lim_{t\rightarrow\infty}(2\Lambda(t+1)-2\Lambda(t)) = \textbf{0}_{N\times d}$ should hold. $(\hat{f}(t),\Lambda(t))$ converges to the stationary point $(\hat{f}^s,\Lambda^s)$. Now show that the stationary point $(\hat{f}^s,\Lambda^s)$ is the optimal point $(\hat{f}^*,\Lambda^*)$.

Take the limit of both sides of \eqref{eq:c_2}\eqref{eq:c_3} yield: 
\begin{eqnarray}
(I+\eta (D+A) \tilde{D}^{-1})(\nabla \hat{O}(\hat{f}^s,D_{all}) +2\Lambda^s)=\textbf{0}_{N\times d}~;\label{eq:thmC24}\\
(D-A)\hat{f}^s = \textbf{0}_{N\times d} \label{eq:thmC25}~.
\end{eqnarray}
Since $I+\eta (D+A) \tilde{D}^{-1}\succ \textbf{0}_{N\times N}$, to satisfy \eqref{eq:thmC24}, $\nabla \hat{O}(\hat{f}^s,D_{all}) +2\Lambda^s=\textbf{0}_{N\times d}$ must hold.

Compare with \eqref{eq:c_4}\eqref{eq:c_5} in Lemma \ref{lemmaP1} and observe that $(\hat{f}^s,\Lambda^s)$ satisfies the optimality condition and is thus the optimal point. Therefore, $(\hat{f}(t),\Lambda(t))$ converges to $(\hat{f}^*,\Lambda^*)$.
\section{Proof of Lemma \ref{lemmaP1}}\label{App_2}
Consider the Private R-ADMM up to $2k$-th iteration. In $(2k-1)$-th iteration, the primal variable is updated via \eqref{eq:P_modify_2}, By KKT condition:
\begin{eqnarray}\label{eq:lemma1}
 \nabla O(f_i(2k-1),D_i) + \epsilon_i(2k-1)=-2\lambda_i(2k-2)\nonumber \\  - \eta \sum_{j \in \mathscr{V}_i}(2f_i(2k-1)-f_i(2k-2)-f_j(2k-2))
\end{eqnarray}
Given $\{f_i(t)\}_{i=1}^N$ for $t\leq 2k-2$, $\{\lambda_i(2k-2)\}_{i=1}^N$ are also given. RHS of \eqref{eq:lemma1} can be calculated completely after releasing $\{f_i(k-1)\}_{i=1}^N$, i.e., the information of $\nabla O(f_i(2k-1),D_i) + \epsilon_i(2k-1)$ is completely released during $(2k-1)$-th iteration. Suppose the Private R-AMDD satisfies $\beta_{2k-1}$-differential privacy during $(2k-1)$ iterations, then in $(2k)$-th iterations, by \eqref{eq:P_modify_3}:
\begin{eqnarray}\label{eq:lemma2}
f_i(2k)=f_i(2k-1) - \frac{1}{2\eta V_i+\gamma}\{\nabla O(f_i(2k-1),D_i) \nonumber \\+\epsilon_i(2k-1) +2\lambda_i(2k-1)\nonumber\\+\eta\sum_{j\in \mathscr{V}_i}(f_i(2k-1)-f_j(2k-1))\}\nonumber
\end{eqnarray}
which is a deterministic mapping taking the outputs from $(2k-1)$-th iteration as input. Because the differential privacy is immune to post-processing \cite{dwork2014algorithmic}, releasing $\{f_i(2k)\}_{i=1}^N$ doesn't increase the privacy loss, i.e., the total privacy loss up to $(2k)$-th iteration can still be bounded by $\beta_{2k-1}$.  

\section{Proof of Theorem \ref{thmP}}\label{App_3}
Use the uppercase letters $X$ and lowercase letters $x$ to denote random variables and the corresponding realizations, and use $\mathscr{F}_{X}(\cdot)$ to denote its probability distribution.


For two neighboring datasets $D_{all}$ and $\hat{D}_{all}$ of the network, by Lemma \ref{lemmaP1}, the total privacy loss is only contributed by odd iterations. Thus, the ratio of joint probabilities (privacy loss) is given by:
\begin{eqnarray}\label{thmP1}
\frac{\mathscr{F}_{F(0:2K)}(\{f(r)\}_{r=0}^2K|D_{all})}{\mathscr{F}_{F(0:2K)}(\{f(r)\}_{r=0}^2K|\hat{D}_{all})} = \frac{\mathscr{F}_{F(0)}(f(0)|D_{all})}{\mathscr{F}_{F(0)}(f(0)|\hat{D}_{all})}  \nonumber\\ \cdot \prod^K_{k=1}\frac{\mathscr{F}_{F(2k-1)}(f(2k-1)|\{f(r)\}_{r=0}^{2k-2},D_{all})}{\mathscr{F}_{F(2k-1)}(f(2k-1)|\{f(r)\}_{r=0}^{2t-2},\hat{D}_{all})}
\end{eqnarray}
Since $f_i(0)$ is randomly selected for all $i$, which is independent of dataset, there is $\mathscr{F}_{F(0)}(f(0)|D_{all}) = \mathscr{F}_{F(0)}(f(0)|\hat{D}_{all})$. First only consider $(2k-1)$-th iteration, since the primal variable is updated according to \eqref{eq:P_modify_2}, by KKT optimality condition:
\begin{eqnarray}\label{thmP2}
\epsilon_i(2k-1)=-\nabla O(f_i(2k-1),D_i) -2\lambda_i(2k-2)\nonumber \\  - \eta \sum_{j \in \mathscr{V}_i}(2f_i(2k-1)-f_i(2k-2)-f_j(2k-2))
\end{eqnarray}
Given $\{f(r)\}_{r=0}^{2k-2}$, $F_i(2k-1)$ and $E_i(2k-1)$ will be bijective $\forall i$, there is:
\begin{eqnarray}\label{thmP3}
\frac{\mathscr{F}_{F(2k-1)}(f(2k-1)|\{f(r)\}_{r=0}^{2k-2},D_{all})}{\mathscr{F}_{F(2k-1)}(f(2k-1)|\{f(r)\}_{r=0}^{2k-2},\hat{D}_{all})}
\nonumber \\= \prod^{N}_{v=1}\frac{\mathscr{F}_{F_v(2k-1)}(f_v(2k-1)|\{f_v(r)\}_{r=0}^{2k-2},D_v)}{\mathscr{F}_{F_v(2k-1)}(f_v(2k-1)|\{f_v(r)\}_{r=0}^{2k-2},\hat{D}_v)}
\nonumber\\= \frac{\mathscr{F}_{F_i(2k-1)}(f_i(2k-1)|\{f_i(r)\}_{r=0}^{2k-2},D_i)}{\mathscr{F}_{F_i(2k-1)}(f_i(2k-1)|\{f_i(r)\}_{r=0}^{2k-2},\hat{D}_i)}
\end{eqnarray}
Since two neighboring datasets $D_{all}$ and $\hat{D}_{all}$ only have at most one data point that is different, the second equality holds is because of the fact that this different data point could only be possessed by one node, say node $i$. Then there is $D_j = \hat{D}_j$ for $j \neq i$.

Given $\{f(r)\}_{r=0}^{2k-2}$, let $g_{k}(\cdot,D_i): \mathbb{R}^d \rightarrow \mathbb{R}^d $ denote the one-to-one mapping from $E_i(2k-1)$ to $F_i(2k-1)$ using dataset $D_i$. 
By Jacobian transformation, there is 
$\mathscr{F}_{F_i(2k-1)}(f_i(2k-1)|D_i) = \mathscr{F}_{E_i(2k-1)}(g^{-1}_{k}(f_i(2k-1),D_i))\cdot|\det(\textbf{J}(g^{-1}_{k}(f_i(2k-1),D_i)))|$
, where $g^{-1}_{k}(f_i(2k-1),D_i)$ is the mapping from $F_i(2k-1)$ to $E_i(2k-1)$ using data $D_i$ as shown in \eqref{thmP2} and $\textbf{J}(g^{-1}_{k}(f_i(2k-1),D_i))$ is the Jacobian matrix of it. Then 
\eqref{thmP1} yields:
\begin{eqnarray}\label{thmP5}
\frac{\mathscr{F}_{F(0:2K)}(\{f(r)\}_{r=0}^{2K}|D_{all})}{\mathscr{F}_{F(0:2K)}(\{f(r)\}_{r=0}^{2K}|\hat{D}_{all})}\nonumber \\= \prod^{K}_{k=1}\frac{\mathscr{F}_{E_i(2k-1)}(g^{-1}_{k}(f_i(2k-1),D_i))}{\mathscr{F}_{E_i(2k-1)}(g^{-1}_{k}(f_i(2k-1),\hat{D}_i))}
\nonumber\\ \cdot \prod^{K}_{k=1} \frac{|\det(\textbf{J}(g^{-1}_{k}(f_i(2k-1),D_i)))|}{|\det(\textbf{J}(g^{-1}_{k}(f_i(2k-1),\hat{D}_i)))|}
\end{eqnarray}
Consider the first part, $E_i(2k-1) \sim \exp\{-\alpha_i(k)||\epsilon||\}$, let $\hat{\epsilon}_i(2k-1) = g^{-1}_{k}(f_i(2k-1),\hat{D}_i)$ and ${\epsilon}_i(2k-1) = g^{-1}_{k}(f_i(2k-1),D_i)$
\begin{eqnarray}\label{thmP6}
\prod^{K}_{k=1}\frac{\mathscr{F}_{E_i(2k-1)}(g^{-1}_{k}(f_i(2k-1),D_i))}{\mathscr{F}_{E_i(2k-1)}(g^{-1}_{k}(f_i(2k-1),\hat{D}_i))}
\nonumber \\= \prod^{K}_{k=1} \exp(\alpha_i(k)(||\hat{\epsilon}_i(2k-1)|| - ||\epsilon_i(2k-1)||))
\nonumber \\\leq \exp(\sum^{K}_{k=1}\alpha_i(k)||\hat{\epsilon}_i(2k-1) - \epsilon_i(2k-1)||)
\end{eqnarray}
Without loss of generality, let $D_i$ and $\hat{D}_i$ be only different in the first data point, say $(x_i^1,y_i^1)$ and $(\hat{x}_i^1,\hat{y}_i^1)$ respectively. 
By \eqref{thmP2}, Assumptions 4 and the facts that $||x_i^n||_2 \leq 1$ (pre-normalization), $y_i^n \in \{+1,-1\}$.
\begin{eqnarray}\label{thmP7}
||\hat{\epsilon}_i(2k-1) - \epsilon_i(2k-1)|| \nonumber \\=||\nabla O(f_i(2k-1),\hat{D}_i)-\nabla O(f_i(2k-1),D_i)||
\leq \frac{2C}{B_i}
\end{eqnarray}

\eqref{thmP6} can be bounded:
\begin{equation}\label{thmP8}
\prod^{K}_{k=1}\frac{\mathscr{F}_{E_i(2k-1)}(g^{-1}_{k}(f_i(2k-1),D_i))}{\mathscr{F}_{E_i(2k-1)}(g^{-1}_{k}(f_i(2k-1),\hat{D}_i))}
 \leq \exp(\sum^{K}_{k=1}\frac{2C\alpha_i(k)}{B_i})
\end{equation}

Consider the second part, the Jacobian matrix $\textbf{J}(g^{-1}_{k}(f_i(2k-1),D_i))$ is:
\begin{eqnarray}\label{thmP10}
\textbf{J}(g^{-1}_{k}(f_i(2k-1),D_i))\nonumber \\ = -\frac{C}{B_i}\sum_{n=1}^{B_i}\mathscr{L}''(y_i^n f_i(2k-1)^T x_i^n)x_i^n(x_i^n)^T
\nonumber \\-\frac{\rho}{N}\nabla^2 R(f_i(2k-1)) - 2\eta V_i\textbf{I}_d\nonumber 
\end{eqnarray}

Define
\begin{eqnarray}
G(k) = \frac{C}{B_i}(\mathscr{L}''(\hat{y}_i^1 f_i(2k-1)^T \hat{x}_i^1)\hat{x}_i^1(\hat{x}_i^1)^T \nonumber \\- \mathscr{L}''(y_i^1 f_i(2k-1)^T x_i^1)x_i^1(x_i^1)^T)\nonumber ~;\\
H(k) = -\textbf{J}(g^{-1}_{k}(f_i(2k-1),D_i))~.\nonumber
\end{eqnarray}

There is:
\begin{eqnarray}\label{thmP11}
 \frac{|\det(\textbf{J}(g^{-1}_{k}(f_i(2k-1),D_i)))|}{|\det(\textbf{J}(g^{-1}_{k}(f_i(2k-1),\hat{D}_i)))|}
\nonumber \\= \frac{|\det(H(k))|}{|\det(H(k)+G(k))|}
 = \frac{1}{|\det(I + H(k)^{-1}G(k))|}
 \nonumber \\ = \frac{1}{|\prod_{j=1}^r(1+\lambda_j(H(k)^{-1}G(k)))|}
\end{eqnarray}

where $\lambda_j(H(k)^{-1}G(k))$ denotes the $j$-th largest eigenvalue of $H(k)^{-1}G(k)$. Since $G(k)$ has rank at most 2, $H(k)^{-1}G(k)$ also has rank at most 2. By Assumptions 4 and 5, the eigenvalue of $H(k)$ and $G(k)$ satisfy
\begin{eqnarray}\label{thmP12}
\lambda_j(H(k)) \geq \frac{\rho}{N} + 2\eta V_i > 0 ~;\nonumber\\
-\frac{Cc_1}{B_i} \leq \lambda_j(G(k)) \leq \frac{Cc_1}{B_i}~.\nonumber
\end{eqnarray}
Implies
\begin{eqnarray}\label{thmP14}
-\frac{c_1}{\frac{B_i}{C}(\frac{\rho}{N}+2\eta V_i)}\leq \lambda_{j}(H(k)^{-1}G(k))
\leq \frac{c_1}{\frac{B_i}{C}(\frac{\rho}{N}+2\eta V_i)}~.\nonumber 
\end{eqnarray}

Since $2c_1 < \frac{B_i}{C}(\frac{\rho}{N}+2\eta V_i)$, there is
\begin{equation*}
-\frac{1}{2}\leq \lambda_{j}(H(k)^{-1}G(k))
\leq \frac{1}{2}.
\end{equation*}

Since $\lambda_{\min}(H(k)^{-1}G(k)) > -1$, there is
\begin{eqnarray}
\frac{1}{|1+\lambda_{\max}(H(k)^{-1}G(k))|^2}  \leq \frac{1}{|\text{det}(I+H(k)^{-1}G(k))|}  \nonumber \\\leq \frac{1}{|1+\lambda_{\min}(H(k)^{-1}G(k))|^2}~. \nonumber
\end{eqnarray}
Therefore, 
\begin{eqnarray}\label{thmP15}
 \prod^{K}_{k=1}\frac{|\det(\textbf{J}(g^{-1}_{k}(f_i(2k-1),D_i)))|}{|\det(\textbf{J}(g^{-1}_{k}(f_i(2k-1),\hat{D}_i)))|}
\nonumber \\ \leq \prod^{K}_{k=1}\frac{1}{|1-\frac{c_1}{\frac{B_i}{C}(\frac{\rho}{N}+2\eta V_i)}|^2}
\nonumber\\ = \exp(-\sum_{k=1}^{K}2\ln(1-\frac{c_1}{\frac{B_i}{C}(\frac{\rho}{N}+2\eta V_i)}))~.
\end{eqnarray}

Since for any real number $x \in [0,0.5]$, $-\ln(1-x)<1.4x$. \eqref{thmP15} can be bounded with a simper expression:
\begin{eqnarray}\label{thmP16}
 \prod^{K}_{k=1}\frac{|\det(\textbf{J}(g^{-1}_{k}(f_i(2k-1),D_i)))|}{|\det(\textbf{J}(g^{-1}_{k}(f_i(2k-1),\hat{D}_i)))|}
\nonumber \\ \leq \exp(\sum_{k=1}^{K}\frac{2.8c_1}{\frac{B_i}{C}(\frac{\rho}{N}+2\eta V_i)})~. 
\end{eqnarray}

Combine \eqref{thmP8}\eqref{thmP16}, \eqref{thmP5} can be bounded:
\begin{eqnarray}
\frac{\mathscr{F}_{F(0:2K)}(\{f(r)\}_{r=0}^{2K}|D_{all})}{\mathscr{F}_{F(0:2K)}(\{f(r)\}_{r=0}^{2K}|\hat{D}_{all})} \nonumber\\ \leq\exp(\sum_{k=1}^{K}\frac{2C}{B_i}(\frac{1.4c_1}{(\frac{\rho}{N}+2\eta V_i)} + \alpha_i(k)))~.
\end{eqnarray}

Therefore, the total privacy loss during $T$ iterations can be bounded by any $\beta$:
\begin{equation*}
 \beta \geq \underset{i \in \mathscr{N}}{\max}\{\sum_{k=1}^{K}\frac{2C}{B_i}(\frac{1.4c_1}{(\frac{\rho}{N}+2\eta V_i)} + \alpha_i(k))\}~.
\end{equation*}

\vfill
\bibliographystyle{IEEEtran}
\bibliography{allerton2018_xueru}

\newpage


\end{document}